\documentclass[11pt]{article}

\usepackage{a4}
\usepackage{algorithmic}
\usepackage{algorithm}
\usepackage{showlabels}

\usepackage{textcomp}
\usepackage{amsmath}
\usepackage{amssymb}
\usepackage{amsthm}
\usepackage{hyperref}
\usepackage{color}
\usepackage{a4}
\hypersetup{
	colorlinks={false},
	linkbordercolor={1 0 0},
	citebordercolor={0 1 0},
}

\theoremstyle{plain}
\newtheorem{theorem}{Theorem}[section]
\newtheorem{lemma}[theorem]{Lemma}

\newtheorem{corollary}[theorem]{Corollary}
\newtheorem{definition}[theorem]{Definition}

\newtheorem{problem}[theorem]{Problem}

\theoremstyle{remark}
\newtheorem*{remark}{Remark}

\newcommand{\E}{\mathbb{E}}

\newcommand{\NP}{\mathsf{NP}}

\newcommand{\fix}{\mathrm{fix}}
\newcommand{\orb}{\mathrm{orb}}
\newcommand{\mov}{\mathrm{move}}

\renewcommand{\angle}[1]{\mbox{$\langle #1 \rangle$}}

\title{The Parameterized Complexity of some Permutation Group Problems}

\author{V. Arvind\\
        The Institute of Mathematical Sciences\\
        C.I.T. Campus\\
        Chennai 600 113, India\\
        arvind@imsc.res.in\\}\date{}

\begin{document}
\maketitle{}

\begin{abstract}
In this paper we study the parameterized complexity of two well-known
permutation group problems which are NP-complete. 

\begin{itemize}
\item Given a permutation group $G=\angle{S}\le S_n$ and a parameter
  $k$, find a permutation $\pi\in G$ such that $|\{i\in [n]\mid
  \pi(i)\ne i\}|\ge k$. This generalizes the $\NP$-complete problem of
  finding a fixed-point free permutation in $G$ \cite{CW10,Lubiw}
  (this is the case when $k=n$). We show that this problem with
  parameter $k$ is fixed parameter tractable. In the process, we give
  a simple deterministic polynomial-time algorithm for finding a fixed
  point free element in a transitive permutation group, answering an
  open question of Cameron \cite{Cam4,CW10}.

\item Next we consider the problem of computing a base for a
  permutation group $G=\angle{S}\le S_n$. A \emph{base} for $G$ is a
  subset $B\subseteq [n]$ such that the subgroup of $G$ that fixes $B$
  pointwise is trivial. This problem is known to be $\NP$-complete
  \cite{Bla}. We show that it is fixed parameter tractable for the
  case of cyclic permutation groups and for permutation groups of
  constant orbit size. For more general classes of permutation groups
  we do not know whether the problem is in FPT or is W[1]-hard.
\end{itemize}

\end{abstract}

\section{Introduction}\label{sec1}

Let $S_n$ denote the group of all permutations on a set of size $n$.
The group $S_n$ is also called the symmetric group of degree $n$. We
refer to a subgroup $G$ of $S_n$, denoted by $G\le S_n$, as a
permutation group (of degree $n$). Let $S\subseteq S_n$ be a subset of
permutations.  The permutation group \emph{generated} by $S$, denoted
by $\angle{S}$, is the smallest subgroup of $S_n$ containing $S$. A
subset $S\subseteq G$ of a permutation group $G$ is a \emph{generating
  set} for $G$ if $G=\angle{S}$. It is easy to see that every finite
group $G$ has a generating set of size $\log_2 |G|$.

Let $G =\angle{S}\le S_n$ be a subgroup of the symmetric group $S_n$,
where $G$ is given as input by a generating set $S$ of
permutations. There are many algorithmic problems on permutation
groups that are given as input by their generating sets (e.g. see
\cite{Sims,FHL80,Luk93,Ser03}). Some of them have efficient
algorithms, some others are $\NP$-complete, and yet others have a
status similar to Graph Isomorphism: they are neither known to be in
polynomial time and unlikely to be NP-complete (unless the
Polynomial-Time Hierarchy collapses). Efficient permutation group
algorithms have played an important role in the design of algorithms
for the Graph Isomorphism problem \cite{Bab79,BKL83}. In fact the
algorithm with the best running time bound for general Graph
Isomorphism is group-theoretic.

We recall some definitions and notions from permutation group theory.
Let $\pi\in S_n$ be a permutation. A \emph{fixed point} of $\pi$ is a
point $i\in [n]$ such that $\pi(i)=i$ and $\pi$ is \emph{fixed point
  free} if $\pi(i)\ne i$ for all $i\in [n]$.

Let $G\le S_n$ and $\Delta\subseteq [n]$ be a subset of the domain.
The \emph{pointwise stabilizer subgroup} of $G$, denoted $G_\Delta$,
is $\{g\in G\mid g(i)=i\textrm{ for all }i\in\Delta\}$.  

A subset $B\subseteq [n]$ is called a \emph{base} for $G$ if the
pointwise stabilizer subgroup $G_B$ is trivial. Thus, if $B$ is a base
for $G$ then each element of $G$ is uniquely determined by its action
on $B$. The problem of computing a base of minimum cardinality is
known to be computationally very useful. Important algorithmic
problems on permutation groups, like membership testing, have nearly
linear time algorithms in the case of small-base groups (e.g. see
\cite{Ser03}). We will discuss the parameterized complexity of the
minimum base problem in Section~\ref{sec3}.

An excellent modern reference on permutation groups is Cameron's book
\cite{Cam1}. Algorithmic permutation group problems are very well
treated in \cite{Luk93,Ser03}. Basic definitions and results on
parameterized complexity can be found in Downey and Fellows' classic
text on the subject \cite{DFbook}. Another, more recent, reference is
\cite{FGbook}.

\section{Fixed point free elements}\label{sec2}

The starting point is the Orbit-Counting lemma. Our discussion will
follow Cameron's book \cite{Cam1}. For each permutation $g\in S_n$ let
$\fix(g)$ denote the number of points fixed by $g$. More precisely,
\[
\fix(g) = |\{i\in [n]\mid g(i)= i\}|.
\]

A permutation group $G\le S_n$ induces, by its action an equivalence
relation on the domain $[n]$: $i$ and $j$ are in the same equivalence
class if $g(i)=j$ for some $g\in G$. Each equivalence class is an
\emph{orbit} of $G$. $G$ is said to be \emph{transitive} if there is
exactly one $G$-orbit. Let $\orb(G)$ denote the number of $G$-orbits
in the domain $[n]$.  We recall the statement.

\begin{lemma}[Orbit Counting Lemma]{\rm\cite{Cam2}}
Let $G\le S_n$ be a permutation group. Then
\begin{eqnarray}\label{orbcount}
\orb(G) = {\frac{1}{|G|}}\sum_{g\in G}\fix(g).
\end{eqnarray}          
I.e.\ the number of $G$ orbits is the average number of
fixed points over all elements of $G$.
\end{lemma}

\begin{proof}
It is useful to recall a proof sketch. Define a $|G|\times n$ matrix
with rows indexed by elements of $G$ and columns by points in $[n]$.
The $(g,i)^{th}$ entry is defined to be $1$ if $g(i)=i$ and $0$
otherwise. Clearly, the $g^{th}$ row has $\fix(g)$ many $1$'s in it.
Let $G_i$ denote the subgroup of $G$ that fixes $i$. The $i^{th}$
column clearly has $|G_i|$ many $1$'s. Counting the number of $1$'s
in the rows and columns and equating them, keeping in mind that
$|G|/|G_i|$ is the size of the orbit containing $i$ yields the
lemma.
\end{proof}

We now recall a theorem of Jordan on permutation groups
\cite{Jo}. See \cite{Serre,Cam4} for very interesting accounts of
it. A permutation group $G\le S_n$ is \emph{transitive} if it has
exactly one orbit.

\begin{theorem}[Jordan's theorem]
If $G\le S_n$ is transitive then $G$ has a fixed point free element.
\end{theorem}

It follows directly from the Orbit counting lemma. Notice that the
left side of Equation~\ref{orbcount} equals $1$. The right side of the
equation is the average over all $\fix(g)$. Now, the identity element
$1$ fixes all $n$ elements. Thus there is at least one element $g\in
G$ such that $\fix(g)=0$. Cameron and Cohen \cite{Cam2} do a more
careful counting and show the following strengthening.

\begin{theorem}{\rm\cite{Cam2}}
If $G\le S_n$ is transitive then there are at least $|G|/n$
elements of $G$ that are fixed point free.
\end{theorem}

We discuss their proof, because we will build on it to obtain our
results. If $G$ is transitive, the orbit counting lemma implies
\[
|G| = \sum_{g\in G}\fix(g).
\]
Take any point $\alpha\in [n]$. We can write the above equation as
\[
|G|=\sum_{g\in G_\alpha}\fix(g) + \sum_{g\in G\setminus G_\alpha}\fix(g).
\]
By the orbit counting lemma applied to the group $G_\alpha$ we have
\[
\sum_{g\in G_\alpha}\fix(g) = \orb(G_\alpha)\cdot |G_\alpha|.
\]

Let $F\subset G$ be the set of all fixed point free elements of
$G$. Clearly, $\sum_{g\in G\setminus G_\alpha}\fix(g)\ge |G\setminus
A|$ as $A\subseteq G\setminus G_\alpha$ and each element of
$G\setminus A$ fixes at least one element. Combining with 
the previous equation we get
\[
|A| \ge \orb(G_\alpha)\cdot |G_\alpha|=\orb(G_\alpha)\cdot \frac{|G|}{n}
\ge \frac{|G|}{n}.
\]

\subsection{The Algorithmic Problem}

We now turn to the problem of computing a fixed point free element in
a permutation group $G\le S_n$ and a natural parameterized version.

As observed by Cameron and Wu in \cite{CW10}, the result of
\cite{Cam2} gives a simple randomized algorithm to find a fixed point
free element in a transitive permutation group $G\le S_n$, where $G$
is given by a generating set $S$: Using Schreier-Sims polynomial-time
algorithm \cite{Sims} we can compute a \emph{strong generating set}
$S'$ for $G$ in polynomial time. And using $S'$ we can sample
uniformly at random from $G$. Clearly, in $O(n)$ sampling trials we
will succeed in finding a fixed point free element with constant
probability. We will show in the next section that this algorithm can
be \emph{derandomized} to obtain a deterministic polynomial time
algorithm (without using CFSG). This answers an open problem of
Cameron discussed in \cite{CW10,Cam4}.

This result is to be contrasted with the fact that computing fixed
point free elements in nontransitive groups $G\le S_n$ is $\NP$-hard.
The decision problem is shown $\NP$-complete in \cite{CW10}. This is
quite similar to Lubiw's result \cite{Lubiw} that checking if a graph
$X$ has a fixed point free automorphism is $\NP$-complete.

We will now introduce the parameterized version of the problem of
computing fixed point free elements in permutation groups. First we
introduce some terminology. We say that a permutation $\pi$
\emph{moves} a point $i\in [n]$ if $\pi(i)\ne i$.\\

\noindent\textbf{$k$-MOVE Problem}\\

\noindent{INPUT}: A permutation group $G=\angle{S}\le S_n$ given
by generators and a number $k$.\\

\noindent{PROBLEM}: Is there an element $g\in G$ that moves at least
$k$ points. \\

For $k=n$ notice that k-MOVE is precisely the problem of checking if
there is a fixed point free element in $G$. The parameterized version
of the problem is to treat $k$ as parameter. We will show that this
problem is fixed parameter tractable.

Let $\mov(g)$ denote the number of points moved by $g$. We define
two numbers $\fix(G)$ and $\mov(G)$:

\begin{eqnarray*}
\fix(G) & = & |\{i\in[n]\mid g(i)=i\textrm{ for all }g\in G\}|\\
\mov(G) & = & |\{i\in[n]\mid g(i)\ne i\textrm{ for some }g\in G\}|\\
\end{eqnarray*}

I.e.\ $\fix(G)$ is the number of points fixed by all of $G$ and
$\mov(G)$ is the number of points moved by some element of
$G$. Clearly, for all $g\in G$, $\mov(g) = n-\fix(g)$ and $\mov(G) =
n-\fix(G)$. Furthermore, notice that $\orb(G)\le \fix(G)+\mov(G)/2$,
and we have $n-\orb(G)\ge \mov(G)/2$. Let $G=\angle{S}\le S_n$ be an
input instance for the $k$-MOVE problem. Substituting $n-\mov(g)$ for
$\fix(g)$ in Equation~\ref{orbcount} and rearranging terms we obtain
\begin{eqnarray}\label{moveqn}
\mov(G)/2\le n-\orb(G)=\frac{1}{|G|}\sum_{g\in G}\mov(g)=\E_{g\in G}[\mov(g)],
\end{eqnarray}
where the expectation is computed for $g$ picked uniformly at random
from $G$. 

We will show there is a deterministic polynomial time algorithm that
on input $G=\angle{S}\le S_n$ outputs a permutation $g\in G$ such that
$\mov(g)\ge n-\orb(G)\ge \mov(G)/2$. Using this algorithm we will
obtain an FPT algorithm for the $k$-MOVE problem. We require the
following useful lemma about computing the average number of points
moved by uniformly distributed elements from a \emph{coset} contained
in $S_n$.

\begin{lemma}\label{orb-gen}
Let $G\pi \subseteq S_n$ be a coset of a permutation group
$G=\angle{S}\le S_n$, where $\pi\in S_n$. There is a deterministic
algorithm that computes $\E_{g\in G}[\mov(g\pi)]$ in time polynomial
in $|S|$ and $n$.
\end{lemma}

\begin{proof}
We again use a double counting argument. Define a $0$-$1$ matrix with
rows indexed by $g\pi, g\in G$ and columns by $i\in [n]$, whose
$(g\pi,i)^{th}$ entry is $1$ if and only if $g(\pi(i))\ne i$. Thus,
the number of $1$'s in the $i^{th}$ column of the matrix is $|G| -
|\{g\in G\mid g(\pi(i))=i\}|$. Now, $|\{g\in G\mid g(\pi(i))=i\}|$ is
zero if $\pi(i)$ and $i$ are in different $G$-orbits and is $|G_i|$ if
they are in the same orbit. In polynomial time we can compute the
orbits of $G$ and check this condition. Also, the number $|G| -
|\{g\in G\mid g(\pi(i))=i\}| = |G|-|G_i|$ is computable in polynomial
time. Call this number $N_i$. It follows that the total number of
$1$'s in the matrix is $\sum_{i=1}^n N_i$, which is computable in
polynomial time. Since $\sum_{i=1}^n N_i=\sum_{g\in G}\mov(g\pi)$, it
follows that $\frac{1}{|G|}\sum_{g\in G}\mov(g\pi)=\E_{g\in
  G}[\mov(g\pi)]$ can be computed exactly in polynomial time.
\end{proof}

\begin{theorem}\label{derand}
There is a deterministic polynomial-time algorithm that takes as input
a permutation group $G=\angle{S}\le S_n$ given by generators and a
permutation $\pi\in S_n$ and computes an element $g\in G$ such that
$\mov(g\pi)\ge \E_{g\in G}[\mov(g\pi)]$.
\end{theorem}

\begin{proof}
We have
\[
\frac{1}{|G|}\sum_{g\in G}\mov(g\pi)=\E_{g\in G}[\mov(g\pi)]=\mu,
\]
and by Lemma\ref{orb-gen} we can compute $\mu$ in polynomial time. We
can write $G$ as a disjoint union of cosets $G=\bigcup_{i=1}^rG_1g_i$,
where $G_1$ is the subgroup of $G$ that fixes $1$ and $g_i$ are the
coset representatives, where the number of cosets $r\le n$. Using
Schreier-Sims algorithm \cite{Sims} we can compute all coset representatives
$g_i$ and a generating set for $G_1$ from the input in polynomial
time.

Now, we can write the summation $\frac{1}{|G|}\sum_{g\in G}\mov(g\pi)$ as
a sum over the cosets $G_1g_i\pi$ of $G_1$: 
\[
\frac{1}{|G|}\sum_{g\in G}\mov(g\pi)= 
\frac{1}{|G|}\sum_{i=1}^r\sum_{g\in G_1}\mov(gg_i\pi). 
\]
For $1\le i\le r$ let 
\[
\mu_i=\frac{1}{|G_1|}\sum_{g\in G_1}\mov(gg_i\pi).
\]
Since $|G|/|G_1|=r$, it follows that
$\mu=\frac{1}{r}\sum_{i=1}^r\mu_i$ is an average of the $mu_i$. Let
$\mu_t$ denote $\max_{1\le i\le r} \mu_i$. Clearly, $\mu\le mu_t$ and
therefore there is some $g\in G_1g_t\pi$ such that $\mov(g)\ge
\mu_t\ge \mu$ and we can continue the search in the coset $G_1g_t$
since we can compute all the $\mu_i$ in polynomial time by
Lemma~\ref{orb-gen}. Continuing thus for $n-1$ steps, in polynomial
time we will obtain a coset $G_{n-1}\tau$ containing the unique
element $\tau$ such that $\mov(\tau)\ge \mu$. This completes
the proof.
\end{proof}

Cameron, in \cite{CW10} and in the lecture notes \cite{Cam4}, raises
the question whether the randomized algorithm, based on uniform
sampling, for finding a fixed point free element in a transitive
permutation group (given by generators) can be derandomized. In
\cite{CW10} a deterministic algorithm (based on the classification of
finite simple groups) is outlined. The algorithm does a detailed case
analysis based on the CFSG and is not easy to verify. Here we show
that the randomized algorithm can be easily derandomized yielding a
simple polynomial-time algorithm. The derandomization is essentially a
simple application of the ``method of conditional probabilities''
\cite{ES,Rag}.

\begin{corollary}
Given a transitive permutation group $G=\angle{S}\le S_n$ by a
generating set $S$, we can compute a fixed point free element
of $G$ in deterministic polynomial time.
\end{corollary}

\begin{proof}
Notice that $\E_{g\in G}[\mov(g)]=n-1$ to begin with. However, since
$G_1$ has at least two orbits, we have by orbit counting lemma that
$\E_{g\in G_1}[\mov(g)]\le n-2$. Hence, for some coset $G_1g_i$ of
$G_1$ in $G$ we must have $\E_{g\in G_1g_i}[\mov(gg_i)] > n-1$.  The
polynomial-time algorithm of Theorem~\ref{derand} applied to $G$ will
therefore continue the search in cosets where the expected value is
strictly more than $n-1$ which means that it will finally compute a
fixed point free element of $G$.
\end{proof}

Given $G=\angle{S}\le S_n$ there is a trivial exponential time
algorithm for finding a fixed point free element in $G$: compute a
strong generating set for $G$ in polynomial time \cite{Sims}. Then
enumerate $G$ in time $|G|.n^{O(1)}$ using the strong generating set,
checking for a fixed point free element. This algorithm could have
running time $n!$ for large $G$.  We next describe a $2^nn^{O(1)}$
time algorithm for finding a fixed point free element based on
inclusion-exclusion and coset intersection.

\begin{theorem}\label{fpfsearch}
Given a permutation group $G=\angle{S}\le S_n$ and $\pi\in S_n$ there
is a $2^{n+O(\sqrt{n}\lg n)}n^{O(1)}$ time algorithm to test if the
coset $G\pi$ has a fixed point free element and if so compute it.
\end{theorem}

\begin{proof}
For each subset $\Delta\subseteq [n]$ we can compute the pointwise
stabilizer subgroup $G_\Delta$. This will take time $2^nn^{O(1)}$
overall. For each $i\in [n]$, let $(G\pi)_i$ denote the subcoset of
$G\pi$ that fixes $i$. Indeed, 
\[
(G\pi)_i = \{g\pi\mid g\in G, g\pi(i)=i\}=G_{\pi(i)}\tau_i\pi,
\] 
if there is a $\tau_i\in G$ such that $\tau_i(\pi(i))=i$ and
$(G\pi)_i=\emptyset$ otherwise. 

Clearly, $G\pi$ has a fixed point free element if and only if the
union $\bigcup_{i=1}^n(G\pi)_i$ is a \emph{proper} subset of $G\pi$.
I.e. we need to check if $|\bigcup_{i=1}^n(G\pi)_i| < |G\pi|=|G|$.
Now, $|\bigcup_{i=1}^n(G\pi)_i|$ can be computed in
$2^{n+O(\sqrt{n}\lg n)}n^{O(1)}$ time using the inclusion exclusion
principle: there are $2^n$ terms in the inclusion-exclusion
formula. Each term is the cardinality of a coset intersection of the
form $\bigcap_{i\in I}(G\pi)_i$, for some subset of indices
$I\subseteq [n]$, which can be computed in time $n^{O(\sqrt{n})}$ time
\cite{BKL83}. Hence, we can decide in $2^{n+O(\sqrt{n}\lg n)}n^{O(1)}$
time whether or not $G\pi$ has a fixed point free element. Notice that
this fixed point free element must be in one of the $n-1$ subcosets of
$G\pi$ that maps $1$ to $j$ for $j\in\{2,3,\ldots,n\}$. The subcoset
of $G\pi$ mapping $1$ to $j$ can be computed in polynomial time
\cite{Sims}. Then we can apply the inclusion exclusion principle to
each of these subcosets, as explained above, to check if it contains a
fixed point free element and continue the search in such a
subcoset. Proceeding thus for $n-1$ steps we will obtain a fixed point
free element in $G\pi$, if it exists, in $2^{n+O(\sqrt{n}\lg
  n)}n^{O(1)}$ time.
\end{proof}

We now prove the main result of this section.

\begin{theorem}\label{fpt1}
There is a deterministic $2^{2k+O(\sqrt{k}\lg k)}k^{O(1)}+n^{O(1)}$
time algorithm for the $k$-\textrm{MOVE} problem and hence the problem
is fixed parameter tractable. Furthermore, if $G=\angle{S}\le S_n$ is
a ``yes'' instance the algorithm computes a $g\in G$ such that
$\mov(g)\ge k$.
\end{theorem}

\begin{proof}
Let $G=\angle{S}\le S_n$ be an input instance of $k$-MOVE with
parameter $k$. By Equation~\ref{moveqn} we know that $\E_{g\in
  G}[\mov(g)\ge \mov(G)/2$. We first compute $\mov(G)$ in polynomial
  time by computing the orbits of $G$. If $\mov(G)\ge 2k$ then the
  input is a ``yes'' instance to the problem and we can apply
  Theorem~\ref{derand} to compute a $g\in G$ such that $\mov(g)\ge k$
  in polynomial time. Otherwise, $\mov(G)\le 2k$. In that case, the
  group $G$ is effectively a permutation group on a set $\Omega
  \subseteq [n]$ of size at most $2k$. For each subset
  $\Delta\subseteq \Omega$ of size at most $k$, we compute the
  pointwise stabilizer subgroup $G_{\Delta}$ of $G$ in polynomial time
  \cite{Sims}. This will take overall $2^{2k}n^{O(1)}$ time. Now, if
  the input is a ``yes'' instance to $k$-MOVE, some subgroup
  $G_\Delta$ must contain a fixed point free element (i.e.\ fixed
  point free in $\Omega\setminus \Delta$). We can apply the algorithm
  of Theorem~\ref{fpfsearch} to compute this element in time
  $2^{2k+O(\sqrt{k}\lg k)}k^{O(1)}$.
\end{proof}

\begin{remark}
We note from the first few lines in the proof of Theorem~\ref{fpt1}
that the application of Theorem~\ref{derand} is actually a polynomial
time reduction from the given $k$-MOVE instance to an instance for
which $\mov(G)\le 2k$. Given $G=\angle{S}\le S_n$ such that
$\mov(G)\le 2k$, note that $G$ is effectively a subgroup of
$S_{2k}$. We can apply the Schreier-Sims algorithm to compute from $S$
a generating set of size $O(k^2)$ for $G$, therefore yielding a
polynomial time computable, $k^{O(1)}$ size kernel (see \cite{FGbook}
for definition) for the $k$-MOVE problem.
\end{remark}

\section{The parameterized minimum base problem}\label{sec3}

In this section we turn to another basic algorithmic problem on
permutation groups.

\begin{definition}
Let $G\le S_n$ be a permutation group. A subset of points $B\subseteq
[n]$ is called a \emph{base} if the pointwise stabilizer subgroup
$G_B$ of $G$ (subgroup of $G$ that fixes $B$ pointwise) is the
identity.
\end{definition}

Since permutation groups with a small base have fast algorithms for
various problems \cite{Ser03}, computing a minimum cardinality base for
$G$ is very useful. The decision problem is $\NP$-complete. On the
other hand, it has a $\lg\lg n$ factor approximation algorithm
\cite{Bla}.

In this section we study the parameterized version of the problem with
base size as parameter. We are unable to resolve if the general case
is FPT or not, we give FPT algorithms in the case of cyclic
permutation groups and for permutation groups with orbits of size
bounded by a constant.\\

\noindent\textbf{$k$-BASE Problem}\\

\noindent{INPUT}: A permutation group $G=\angle{S}\le S_n$ given
by generators and a number $k$.\\

\noindent{PROBLEM}: Is there a base of size at most $k$ for $G$. The
search version is the find such a base.\\

A trivial $n^{k+O(1)}$ algorithm would cycle through all candidate
subsets $B$ of size at most $k$ checking if $G_B$ is the identity.

\begin{remark}
If the elements of the group $G\le S_n$ are explicitly listed, then
the $k$-BASE problem is essentially a hitting set problem, where the
hitting set $B$ has to intersect, for each $g\in G$, the subset of
points moved by $g$. However, the group structure makes it different
from the general hitting set problem and we do not know how to exploit
it algorithmically in the general case.
\end{remark}

\subsection{Cyclic Permutation Groups}

We give an FPT algorithm for the special case when the input
permutation group $G=\angle{S}$ is cyclic. While this is only a
special case, we note that the minimum base problem is NP-hard even
for cyclic permutation groups \cite[Theorem 3.1]{Bla}.

\begin{theorem}\label{abelfpt}
  The $k{-}\mathrm{BASE}$ problem for cyclic permutation groups is fixed
parameter tractable. 
\end{theorem}

\begin{proof}
  Let $G=\angle{S}\le S_n$ be a cyclic permutation group as instance
  for $k$-BASE. Using known polynomial-time algorithms
  \cite{Sims,Luk93} we can compute a decomposition of $G$ into a
  direct product of cyclic groups of prime power order.
\[
G=H_1\times H_2\times\ldots\times H_\ell
\]
where each $H_i$ is cyclic of prime power order. Let
$H_i=\angle{g_i}$, where the order of $g_i$, $o(g_i)=p_i^{e_i}, 1\le
i\le \ell$, where the $p_i$'s are all distinct. Notice that
$|G|=p_1^{e_1}p_2^{e_2}\ldots p_\ell^{e_\ell}$. We can assume $|G|\le n^k$,
Otherwise, $G$ does not have a size $k$ base and the algorithm can
reject the instance. Since 
\[
(\ell/e)^\ell\le \ell! \le  p_1p_2\cdots p_\ell  \le  n^k,
\]

it follows that $\ell=O(\frac{k\lg n}{\lg \lg n})$.

For each $g_i$, when we express it as a product of disjoint cycles
then the length of each such cycle is a power of $p_i$ that divides
$p_i^{e_i}$, and there is at least one cycle of length $p_i^{e_i}$.
Clearly, any base for $G$ must include at least one point of some
$p_i^{e_i}$-cycle (i.e.\ cycle of length $p_i^{e_i}$) of $g_i$, for
each $i$. Otherwise, the cyclic subgroup $H_i$ of $G$ will not become
identity when the points in the base are fixed. For each index
$i~:~1\le i\le \ell$, define the set of points
\[
S_i = \{\alpha\in[n]\mid \alpha \textrm{ is in some
}p_i^{e_i}~\textrm{cycle of }g_i\}.
\]

\noindent\textbf{Claim.}~ Let $B\subseteq [n]$ be a subset of size
$k$. Then $B$ is a base for $G$ if and only if $B$ is a hitting set
for the collection of sets $\{S_1,S_2,\ldots,S_\ell\}$.\\

\noindent\textbf{Proof of Claim.}~Clearly, it is a necessary
condition. Conversely, suppose $|B|=k$ and $B\cap S_i\ne \emptyset$
for each $i$. Consider the partition of $[n]$ into the orbits of $G$:
\[
[n] = \Omega_1\cup \Omega_2\cup \cdots \cup \Omega_r.
\]

For each $g_i$, a cycle of length $p_i^{e_i}$ in $g_i$ is wholly
contained in some orbit of $G$. Indeed, each orbit of $G$ must be a
union of a subset of cycles of $g_i$. Since $B\cap S_i\ne \emptyset$,
some $p_i^{e_i}$-cycle $C_i$ of $g_i$ will intersect $B$.

Assume, contrary to the claim, that there is a $g\in G_B$ such that
$g\ne 1$. We can write $g=g_1^{a_1}g_2^{a_2}\ldots g_\ell^{a_\ell}$
for nonnegative integers $a_i < p_i^{e_i}$. Suppose $g_j^{a_j}\ne 1$.
Then raising both sides of the equation $g=g_1^{a_1}g_2^{a_2}\ldots
g_\ell^{a_\ell}$ to the power $\frac{|G|}{p_j^{e_j}}$, we have
\[
g'=g^{\frac{|G|}{p_j^{e_j}}}=g_j^{\beta_j},
\]
where $\beta_j<p_j^{e_j}$. Moreover, $\beta_j=
\frac{|G|a_j}{p_j^{e_j}}(mod~ p_j^{e_j})$ is nonzero because $a_j\ne 0(mod~
p_j^{e_j})$ and $|G|/p_j^{e_j}$ does not have $p_j$ as factor.
 
By assumption, some $p_j^{e_j}$-cycle $C_j$ of $g_j$ intersects $B$.
Since $\beta_j$ is nonzero and strictly smaller than $p_j^{e_j}$, none
of the points of $C_j$ are fixed by $g_j^{\beta_j}$ which contradicts
the assumption that $g$ and hence $g'$ is in $G_B$. This proves
the claim.\\

We now explain the FPT algorithm. If $|G| > n^k$ then there is no base
of size $k$. Hence we can assume $|G|\le n^k$. As already observed,
$\ell =O(\frac{k\lg n}{\lg\lg n})$. Thus, we need to solve the
  $k$-hitting set problem for a collection of at most $O(\frac{k\lg
    n}{\lg\lg n})$ many sets $\{S_1,S_2,\ldots,S_\ell\}$. We can think
  of it as a problem of $k$-coloring the indices $\{1,2,\ldots,\ell\}$
  such that for each color class $I$ we have $\cap_{i\in I} S_i\ne
  \emptyset$ and we can pick any one point for each such
  intersection. Notice that there are at most $k^\ell=n^{\frac{k\lg
      k}{\lg\lg n}}$ many such colorings. Now, if $k\lg k\le \lg \lg
  n$ this number is bounded by $n^{O(1)}$ can we can cycle through all
  these $k$-colorings in polynomial time and find a good $k$-coloring
  if it exists. On the other hand, if $k\lg k>\lg \lg n$ then $n^k\le
  2^{k^{k+1}}$ which means the brute force search gives an FPT time
  bound.
\end{proof}

\subsection{Bounded Orbit Permutation Groups}

We give an FPT algorithm for another special case of the $k$-BASE
problem: Let $G=\angle{S}\le S_n$ such that $G$ has orbits of size
bounded by a fixed constant $b$. I.e.\ $[n]=\biguplus_{i=1}^m
\Omega_i$, where $|\Omega_i|\le b$ for each $i$. This is again an
interesting special case as the minimum base problem is NP-hard even
for orbits of size bounded by $8$ \cite[Theorem 3.2]{Bla}.

Suppose $G$ has a base $B=\{i_1,i_2,\ldots,i_k\}$ of size $k$. Then
$G$ has a pointwise stabilizer tower $G=G_0\ge G_1 \ge\ldots\ge
G_k=\{1\}$ obtained by successively fixing the points of $B$. More
precisely, $G_j$ is the subgroup of $G$ that pointwise fixes
$\{i_1,i_2,\ldots,i_j\}$. Now, $\frac{|G_{j-1}|}{|G_j|}$ is the orbit
size of the point $i_j$ in the group $G_{j-1}$. Furthermore, $b$ is
also a bound on this orbit size. Therefore, $|G|\le b^k$. Hence in
$b^kn^{O(1)}$ time we can list all elements of $G$. Let
$G=\{g_1,g_2,\ldots,g_N\}$, where $N\le b^k$, where $g_1$ is the
identity element.

For each $g_i\in G, i\ge 2$, let $S_i=\{j\in[n]\mid g_i(j)\ne j\}$
denote the nonempty subset of points not fixed by $g_i$. Then a subset
$B\subset [n]$ of size $k$ is a base for $G$ if and only if $B$ is a
hitting set for the collection $S_2,S_3,\ldots,S_N$. The next claim
is straightforward. 

\noindent\textbf{Claim.}~~There is a size $k$ hitting set contained in
$[n]$ for the sets $\{S_2,S_3\ldots,S_N\}$ if and only if there is a
  partition of $\{2,3,\ldots,N\}$ into $k$ parts $I_1,I_2,\ldots,I_k$
    such that $\cap_{j\in I_r}S_j\ne \emptyset$ for each
    $r=1,2,\ldots,k$.

As $N\le b^k$, the total number of $k$-partitions of
$\{2,3,\ldots,N\}$ is bounded by $k^N\le k^{b^k}$. We can generate
  them and check if any one of them yields a hitting set of size $k$
  by checking the condition in the above claim. The overall time taken
  by the algorithm is given by the FPT time bound $k^{b^k}n^{O(1)}$.
We have shown the following result.

\begin{theorem}
Let $G=\angle{S}\le S_n$ such that $G$ has orbits of size bounded by
$b$, be an instance for the $k$-BASE problem with $k$ as parameter.
Then the problem has an FPT algorithm of running time
$k^{b^k}n^{O(1)}$.
\end{theorem}

\section{Concluding Remarks}

The impact of parameterized complexity on algorithmic graph theory
research, especially its interplay with graph minor theory, has been
very fruitful in the last two decades. This motivates the study of
parameterized complexity questions in other algorithmic problem
domains like, for example, group-theoretic computation. To this end,
we considered parameterized versions of two well-known classical
problems on permutation groups. We believe that a similar study of
other permutation group problems can be a worthwhile direction.

\newcommand{\etalchar}[1]{$^{#1}$}
\providecommand{\bysame}{\leavevmode\hbox to3em{\hrulefill}\thinspace}
\providecommand{\MR}{\relax\ifhmode\unskip\space\fi MR }
\providecommand{\MRhref}[2]{%
  \href{http://www.ams.org/mathscinet-getitem?mr=#1}{#2}
}
\providecommand{\href}[2]{#2}

\providecommand{\ON}{\ensuremath{\mathcal{O}}}
  \providecommand{\class}[1]{\textsf{\upshape{#1}}}
  \providecommand{\problem}[1]{\textsc{#1}}
  \csname@ifundefined\endcsname{iflongjournalnames}{\newif
  \iflongjournalnames\longjournalnamestrue}\relax
  \providecommand{\journalname}[2]{\iflongjournalnames #1\else #2\fi}
  \providecommand{\proceedings}[3]{\iflongjournalnames Proceedings of #1 #2
  (#3)\else Proc. #1 #3\fi}

\end{document}